\RequirePackage[l2tabu,orthodox]{nag}
\documentclass
[11pt,letterpaper]
{article} 

\usepackage[notes=true,later=false,camera=false]{dtrt}
\usepackage[utf8]{inputenc}
\usepackage{ stmaryrd }
\usepackage{xspace,enumerate}
\usepackage[T1]{fontenc}
\usepackage[full]{textcomp}
\usepackage[american]{babel}
\usepackage{mathtools}

\renewcommand{\hat}[1]{\widehat{#1}}
\usepackage{amsthm}
\usepackage{thmtools}
\usepackage[capitalise,nameinlink]{cleveref}

\usepackage{empheq}

\definecolor{citeblue}{HTML}{0055cc}
\hypersetup{
colorlinks=true,
urlcolor=citeblue,
linkcolor=NavyBlue,
citecolor=PineGreen,
linktocpage=true,
}
\renewcommand*\backref[1]{\ifx#1\relax \else (pg. #1) \fi}

\renewcommand{\tilde}{\widetilde}

\usepackage{paralist}
\usepackage{turnstile}
\usepackage[framemethod=TikZ]{mdframed}
\mdfsetup{frametitlealignment=\center}
\usepackage{tikz}
\usepackage{caption}
\DeclareCaptionType{Algorithm}
\usepackage{newfloat}

\newtheorem{theorem}{Theorem}[section]
\newtheorem{lemma}[theorem]{Lemma}
\newtheorem*{lemma*}{Lemma}
\newtheorem{claim}[theorem]{Claim}
\newtheorem{proposition}[theorem]{Proposition}
\newtheorem{fact}[theorem]{Fact}
\newtheorem{corollary}[theorem]{Corollary}

\theoremstyle{definition}
\newtheorem{definition}[theorem]{Definition}
\newtheorem*{definition*}{Definition}

\usepackage[linesnumbered,ruled]{algorithm2e}
\crefname{lemma}{Lemma}{Lemmas}
\crefname{fact}{Fact}{Facts}
\crefname{theorem}{Theorem}{Theorems}
\crefname{mtheorem}{Theorem}{Theorems}
\crefname{itheorem}{Theorem}{Theorems}
\crefname{corollary}{Corollary}{Corollaries}
\crefname{claim}{Claim}{Claims}
\crefname{example}{Example}{Examples}
\crefname{algorithm}{Algorithm}{Algorithms}
\crefname{problem}{Problem}{Problems}
\crefname{definition}{Definition}{Definitions}
\crefname{equation}{Eq.}{Eq.}
\crefname{strategy}{Strategy}{Strategies}
\crefname{observation}{Observation}{Observations}

\Crefname{algocf}{Algorithm}{Algorithms}

\usepackage[
letterpaper,
top=1.2in,
bottom=1.2in,
left=1in,
right=1in]{geometry}
\usepackage{newpxtext} 
\usepackage{textcomp} 
\usepackage[scr=rsfso]{mathalfa}
\usepackage{bm} 

\usepackage{microtype}

\usepackage{footnotebackref}

\newcommand{\AND}{\operatorname{AND}}

\allowdisplaybreaks

\newcommand{\R}{{\mathbb R}}
\newcommand{\N}{{\mathbb N}}

\newcommand{\eps}{\varepsilon}

\newcommand{\E}{{\mathbb E}}

\newcommand{\ip}[1]{\langle #1 \rangle}

\newcommand{\seq}{\subseteq}

\newcommand{\cK}{\mathcal K}

\newcommand{\sym}{\operatorname{sym}}

\renewcommand{\emptyset}{\varnothing}
\renewcommand{\geq}{\geqslant}
\renewcommand{\ge}{\geqslant}
\renewcommand{\leq}{\leqslant}
\renewcommand{\le}{\leqslant}

\renewcommand{\epsilon}{\varepsilon}

\newcommand{\ignore}[1]{}

\newcommand{\Vol}{\mathrm{Vol}}

\usepackage{titling}
\thanksmarkseries{arabic}

\begin{document}

\title{Optimal Sparsifiers for Minkowski Sums and Sums of Seminorms}

\date{\today}

\author{
Arpon Basu\thanks{Princeton University. \href{mailto:arpon.basu@princeton.edu}{\texttt{arpon.basu@princeton.edu}}} 
\and
Joshua Brakensiek\thanks{University of California, Berkeley. Supported in part by a Simons Investigator award of Venkatesan Guruswami, and NSF awards CCF-2211972 and DMS-2503280 \href{mailto:josh.brakensiek@berkeley.edu}{\texttt{josh.brakensiek@berkeley.edu}}} 
\and
Yeyuan Chen\thanks{Department of EECS, University of Michigan, Ann Arbor. Supported in part by National Science Foundation grant No. CCF-2236931. 
 \href{mailto:yeyuanch@umich.edu}{\texttt{yeyuanch@umich.edu}}}  
\and
Aaron Putterman\thanks{School of Engineering and Applied Sciences, Harvard University, Cambridge, Massachusetts, USA. Supported in part by Simons Investigator Awards to Madhu Sudan and Salil Vadhan, a Jane Street Graduate Research Fellowship, and AFOSR award FA9550-25-1-0112. \href{mailto:aputterman@g.harvard.edu}{\texttt{aputterman@g.harvard.edu}}}
\and
Victor Reis\thanks{Microsoft Research, Redmond. \href{mailto:victorol@microsoft.com}{\texttt{victorol@microsoft.com}}}
\and
Zihan Zhang\thanks{Department of Computer Science and Engineering, The Ohio State University. \href{mailto:zhang.13691@buckeyemail.osu.edu}{\texttt{zhang.13691@osu.edu}}}  
}

\maketitle

\begin{abstract}
We extend the recent work of Reis and Rothvoss on sparsifying sums of $\ell_1$ norms to the more general task of sparsifying (Minkowski) sums of centrally symmetric, convex sets. As our main result, we prove that for any $\eps > 0$ and centrally symmetric, convex sets $C_1, \ldots, C_m\seq\R^n$ there is a choice of weights $\lambda_1, \dots , \lambda_m \in \R_{\geq 0}$ such that at most $O(n / \eps^2)$ of the weights are non-zero, and 
\[(1 - \eps)\cdot C\seq\sum_{i = 1}^m\lambda_i\cdot C_i\seq(1 + \eps)\cdot C,\]
where $C:= C_1 + \cdots + C_m$ refers to the Minkowski sums of the sets $C_1, \ldots, C_m$, and $\lambda\cdot C$ refers to the dilation of the set $C$.

As immediate applications of this result, we obtain sparsifiers of size $O(n / \eps^2)$ for sparsifying sums of seminorms in $n$-dimensional space, improving on the $O\left ( \frac{n \log(n/\eps) \cdot \log^{2.5}(n)}{\eps^2} \right )$ size sparsifiers from the work of Jambulapati, Lee, Liu, and Sidford (FOCS 2023). This further yields optimal size hypergraph cut sparsifiers with $O(n / \eps^2)$ hyperedges, improving on the $O(n \log(n) / \eps^2)$ size sparsifiers from the work of Chen, Khanna, and Nagda (FOCS 2020). More generally, this also gives optimal size sparsifiers for sums of symmetric submodular functions.
\end{abstract}

\pagenumbering{gobble}

\clearpage

\pagebreak

\tableofcontents

\pagebreak

\pagenumbering{arabic}

\section{Introduction}

Sparsification refers to the general procedure by which a large, complex object is replaced by an ideally much smaller object which nevertheless preserves certain properties of interest. Such sparsification was first introduced (independently) in the work of Bencz\'ur and Karger \cite{BenczurK96} and Talagrand \cite{talagrand1990embedding}. \cite{BenczurK96} studied the so-called cut sparsification of \emph{graphs} $G = (V, E, w)$ where the goal is to select a subset of the \emph{edges} of the graph, while still preserving all of the cut-sizes. \cite{talagrand1990embedding} instead studied a continuous generalization of this notion called $\ell_1$-norm sparsification, where one is provided with a sum of $\ell_1$ norms of linear functions, and wishes to choose a subset which preserves the sum on every input. 

Since these seminal works, sparsification has evolved into a rich field. The works of Spielman and Teng \cite{SpielmanT11}, Spielman and Srivastava \cite{SpielmanS11}, and Batson, Spielman, and Srivastava \cite{BatsonSS14} generalized cut sparsification of graphs to a stronger notion called \emph{spectral sparsification}. These notions of graph sparsification have proved to be useful in a plethora of algorithmic applications, ranging from graphs  \cite{BenczurK96, spielman2004nearly, SpielmanS11}, to linear system solving \cite{spielman2004nearly}, and numerous problems in the graph streaming literature \cite{AhnGM12, AhnGM12b, mcgregor2014graph,AbrahamDKKP16, kapralov2017single, kapralov2019faster}.

Other works have sought to understand the theoretical limits of \emph{which} objects can be sparsified. This exploration began with the study of \emph{hypergraph} cut sparsification and CSP sparsification in the work of \cite{KoganK15}, which led to works improving sparsifier sizes \cite{ChenKN20}, a study of spectral hypergraph sparsification \cite{SomaY19, KapralovKTY21, KapralovKTY21a, Lee23, JambulapatiLS23}, sparsification of codes \cite{KhannaPS24, KhannaPS25, BrakensiekG25, BGP26, khanna2026unified, lovett2026moonflowers}, as well as other generalizations like sparsifying sums of norms \cite{jambulapati2023sparsifying}, and sparsifying sums of PSD matrices \cite{HsiehLMPZ26, BasuKLM26, BKMT26}.

Despite this vast attention devoted to sparsification, there is one common thread which plagues \emph{almost} all of these works: random sampling. Indeed, in all instances of sparsification via random sampling, there is an inherent blow-up (in the size of the sparsifier) which scales as $\log(\dim)$. For instance, in the cut-sparsification of graphs of \cite{BenczurK96}, this led to sparsifiers preserving $O_{\eps}(|V| \log(|V|))$ many edges, and in the work of \cite{talagrand1990embedding}, this led to $\ell_1$-norm sparsifiers preserving $O_{\eps}(n \log n)$ many terms. Removing these extra logarithmic factors remained an open question for the ensuing years. Only in the work of Batson, Spielman, and Srivastava \cite{BatsonSS14} was this first resolved for cut (and in fact, the stronger notion of spectral) sparsification, whereby they produced sparsifiers which retained only $O_{\eps}(n)$ many edges. For $\ell_1$ norms, this remained open much longer, until the very recent work of Reis and Rothvoss \cite{RR26}, which produced $\ell_1$-norm sparsifiers with only $O_{\eps}(n)$ many terms. Along the way, the only other distinct instances of sparsification that avoid this inherent logarithmic factor are due to De Carli Silva, Harvey and Sato \cite{silva2015sparse} on sparsifying sums of PSD matrices (and hypergraph cut sparsification for arity $r = 3$), and the (even more) recent work of Basu, Kothari, Meka and Tudose~\cite{BKMT26} which extends \cite{RR26} to sparsifying Abelian Cayley graphs. 

In this work, we re-visit the framework of Reis and Rothvoss \cite{RR26} and ask just how far it can be pushed in constructing optimal size sparsifiers.

\subsection{Our Results}

As our main result, we show that the framework of \cite{RR26} can be extended to imply optimal sparsification for sums of \emph{support functions of centrally symmetric, compact, convex sets} a.k.a., symmetric convex bodies\footnote{Some definitions of convex body require the convex body to have an interior point. We follow the presentation of \cite{Gruber07,Schneider14} where convex bodies need to only be non-empty with a \emph{proper} convex body having non-empty interior.}. More formally, one can consider the space $E = \R^n$. A set $C_i \subseteq E$ is said to be:
\begin{enumerate}
    \item \emph{Centrally symmetric} if whenever $x \in C_i$, then $-x \in C_i$.
    \item \emph{Compact} if $C_i$ is closed and bounded.
    \item \emph{Convex} if for any $x, y \in C_i$ and $\lambda \in [0,1]$, $\lambda \cdot x + (1 - \lambda) \cdot y \in C_i$.
\end{enumerate}
For such a set $C_i$, the support function is defined as $h_{C_i}(u) = \max_{x \in C_i} \langle x, u \rangle$ for $u \in E$. As our main theorem, we show the following:

\begin{theorem}[Sparsifying Sums of Support Functions]\label{thm:main-intro}
There is a universal constant $C_{\star}>0$ with the following property.  Let $E$ be a finite-dimensional real vector space, and let $C_1,\dots,C_m\subseteq E$ be compact centrally symmetric convex sets. Let
\[
  C=C_1+\cdots+C_m,
  \qquad
  d=\dim\mathrm{span}(C),
\]
be the Minkowski sum of the convex bodies.
For every $0<\eps<1$, there exist coefficients $\lambda_1,\dots,\lambda_m\ge 0$, supported on at most $C_{\star} \cdot d/\eps^2$ indices, such that for every $u \in E$,
\[
  (1-\eps) \cdot \sum_{i = 1}^m h_{C_i}(u)
  \leq\sum_{i = 1}^m \lambda_i \cdot h_{C_i}(u)
  \leq (1+\eps) \cdot \sum_{i = 1}^m h_{C_i}(u).
\]
Equivalently, we have the following containment of convex bodies:
\[
    (1 - \eps) \cdot C \subseteq \sum_{i=1}^m \lambda_i \cdot C_i \subseteq (1 + \eps) \cdot C.
\]

\end{theorem}

Importantly, the size of the sparsifier scales as $O(d / \eps^2)$ \emph{without} any logarithmic factors depending on $d$ or $n$. This is the key improvement with respect to prior work. To appreciate the expressivity of \cref{thm:main-intro}, we consider a few examples. First, as mentioned above, our work hinges on a generalization of the framework of \cite{RR26}; it is therefore natural to ensure that \cref{thm:body-main} generalizes their main theorem. In their work they consider functions of the form $f_i: \R^n \rightarrow \R_{\geq 0}$ where $f_i(u) = |\langle a_i, u \rangle|$ for some fixed $a_i \in \R^n$. Given such a collection of functions $f_1, \dots , f_m$ and a parameter $\eps > 0$, the goal is then to find weights $\lambda_1, \dots , \lambda_m \in \R_{\geq 0}$ such that \emph{for every} $u \in \R^n$,
\[
(1 - \eps) \cdot \sum_{i = 1}^m f_i(u) \leq \sum_{i = 1}^m \lambda_i \cdot f_i(u) \leq (1 + \eps) \cdot \sum_{i = 1}^m f_i(u).
\]
Reis and Rothvoss \cite{RR26} show that such sparsifiers can be built while having only $O(n / \eps^2)$ many non-zero weights, which improved upon the previous best of $O(n \log(n) / \eps^2)$ from the work of \cite{talagrand1990embedding}. We can now observe that this fits \emph{exactly} into our framework above: indeed, for each function $f_i$, one simply considers the convex set which is $C_i = \{\alpha \cdot a_i: \alpha \in [-1,1]\}$. In this way, the support function becomes $h_{C_i}(u) = \max_{x \in C_i} \langle x, u \rangle = \max (\langle a_i, u \rangle, \langle -a_i, u \rangle) = |\langle a_i, u \rangle|$.

Thus, the work of \cite{RR26} corresponds exactly with the setting in which each $C_i$ above is a \emph{one-dimensional} convex set. The key expressivity of \cref{thm:main-intro} instead comes from the fact that it works for \emph{arbitrary dimensional} convex sets (again, provided they are compact and centrally symmetric). 

\paragraph{Sparsifying Sums of Seminorms} To demonstrate this expressivity, we consider the sparsification of sums of seminorms. In this setting, first introduced in the work of Jambulapati, Lee, Liu, and Sidford~\cite{jambulapati2023sparsifying}, one is provided with a collection of seminorms $f_1, \dots , f_m: \R^n \rightarrow \R_{\geq 0}$. A function $f$ is a seminorm if it satisfies:
\begin{enumerate}[(1)]
        \item \parhead{Positive Homogeneity} For all $x \in \R^n$ and $\lambda \in \R$, $f(\lambda x) = |\lambda| \cdot f(x)$.
        \item \parhead{Subadditivity} For all $x, y \in \R^n$, $f(x + y) \leq f(x) + f(y)$.
\end{enumerate}
Crucially, \emph{all} seminorms can be written as support functions of centrally symmetric, compact, convex sets. By using this identity in combination with \cref{thm:main-intro}, we obtain the following theorem on sparsifying sums of seminorms:

\begin{theorem}\label{thm:intro-seminorm-sparsifiers}
    Let $f_1, \dots , f_m: \R^n \rightarrow \R_{\geq 0}$ be a collection of seminorms. Then, for any $\eps > 0$, there exists a set of weights $\lambda_1, \dots , \lambda_m \in \R_{\geq 0}$ with only $O(n / \eps^2)$ of the weights non-zero such that for every $x \in \R^n$:
    \[
    (1 - \eps) \cdot \sum_{i = 1}^m  f_i(x) \leq \sum_{i = 1}^m \lambda_i \cdot f_i(x) \leq  (1 + \eps) \cdot \sum_{i = 1}^m  f_i(x).
    \]
\end{theorem}
The previous best known sparsifier size for sums of seminorms was due to the aforementioned work of \cite{jambulapati2023sparsifying}, which retained $O\left ( \frac{n \log(n/\eps) \cdot \log^{2.5}(n)}{\eps^2} \right )$ many terms in the sparsifier. Our work shaves off all of these logarithmic factors thereby \emph{settling} the optimal size complexity of seminorm sparsification. Indeed, because graph cut-sparsification can be captured by sparsifying sums of seminorms, the work of Carlson, Kolla, Srivastava, and Trevisan \cite{CarlsonKST19} implies that there are sums of seminorms for which $(1 \pm \eps)$ sparsification requires the retention of $\Omega(n / \eps^2)$ many terms.

\paragraph{Sparsifying Hypergraphs} Seminorms themselves already generalize many frequently studied sparsification problems. As our first application, we are able to conclude optimal size \emph{hypergraph cut} sparsifiers. In this setting, one is provided with a hypergraph $H = (V, E, w)$, where $w: E \rightarrow \R_{\geq 0}$ is a collection of non-negative hyperedge weights. The goal is to compute a re-weighted sub-hypergraph $H' = (V, E, w')$ with $\mathrm{Supp}(w')$ much smaller than $|E|$ which nevertheless preserves all \emph{cut-sizes} in the hypergraph. In this context, for a set $S \subseteq V$, the cut-size of $S$ is defined as
    \[
    \delta_H(S) := \sum_{e \in E} w(e) \cdot \mathbf{1}[S \cap e \neq \emptyset \text{ and } (V-S) \cap e \neq \emptyset],
    \]
intuitively capturing the total weight of hyperedges which are crossing from $S$ to $V - S$. A hypergraph $H'$ defined on the same vertex set is then a $(1 \pm \eps)$ cut sparsifier of $H$ if for every $S \subseteq V$, $\delta_{H'}(S) \in (1 \pm \eps) \cdot \delta_H(S)$.

This cut-sparsification was first proposed in the work of Kogan and Krauthgamer \cite{KoganK15}, where they showed there exist $(1 \pm \eps)$ sparsifiers which retain only $O(n \cdot (r + \log(n)) / \eps^2)$ many hyperedges, where $n = |V|$ and $r$ is the maximum size of any hyperedge. This was later improved in the work of Chen, Khanna, and Nagda \cite{ChenKN20} to size $O(n \log(n) / \eps^2)$, which has remained the best ever since, despite extensive follow-up attention \cite{SomaY19, KapralovKTY21, KapralovKTY21a, Lee23, JambulapatiLS23, jambulapati2023sparsifying, KhannaPS24, quanrud2024quotient}.

Nevertheless, because hypergraph cut functions have a natural extension as seminorms, we immediately obtain the following theorem:

\begin{theorem}\label{thm:hypergraph-intro}
    Let $H = (V, E, w)$ be any hypergraph (with non-negative weights). Then, for any $\eps > 0$, there exists a hypergraph $H' = (V, E, w')$ (again, with non-negative weights) such that:
    \begin{enumerate}[(1)]
        \item $H'$ is a $(1 \pm \eps)$ hypergraph cut-sparsifier of $H$. 
        \item $H'$ has at most $O(|V| / \eps^2)$ hyperedges with non-zero weight: i.e., $ \left | \{ e \in E: w'(e) > 0\}\right | \leq O(|V| / \eps^2)$.
    \end{enumerate}
\end{theorem}

This removes the final logarithmic factor from the work of \cite{ChenKN20}, and establishes $O(|V| / \eps^2)$ as the \emph{optimal} size for any hypergraph cut sparsifiers (again, with the matching lower bound owed to \cite{CarlsonKST19}).

\paragraph{Sparsifying Sums of Submodular Functions}

Another application of our seminorm sparsification is to sparsifying sums of submodular functions. Recall that given an integer $n$, a function $f: 2^{[n]} \rightarrow \R_{\geq 0}$ is said to be submodular if for every $S, T \subseteq [n]$
\[
f(S) + f(T) \geq f(S \cap T) + f(S \cup T). 
\]
A well-studied question in sparsification (see, e.g., \cite{rafiey2022sparsification, kudla2023sparsification, kenneth2023cut, jambulapati2023sparsifying,khanna2024almost, bao2026sparsify}) asks whether, given a collection of submodular functions $f_1, \dots , f_m: 2^{[n]} \rightarrow \R_{\geq 0}$, one can compute a smaller $(1 \pm \eps)$ sparsifier. In this context, mirroring above, a sparsifier is a collection of weights $\lambda_1, \dots , \lambda_m \in \R_{\geq 0}$ such that for every $S \subseteq [n]$ it is the case that 
    \[
    (1 -\eps) \cdot \sum_{i = 1}^m f_i(S) \leq \sum_{i =1}^m \lambda_i \cdot f_i(S) \leq (1 + \eps) \cdot \sum_{i = 1}^m f_i(S).
    \]
As observed in \cite{jambulapati2023sparsifying}, whenever the submodular functions $f_i$ are additionally \emph{symmetric} (meaning that $f_i(S) = f_i( [n] - S)$) a crucial result of Lov\'asz \cite{lovasz1983submodular} implies that such submodular functions can be written as continuous seminorms. In combination with \cref{thm:intro-seminorm-sparsifiers}, we then obtain the following theorem:

\begin{theorem}\label{thm:intro-symmetric-sparsification}
    Let $n$ be an integer, and let $f_1, \dots , f_m: 2^{[n]} \rightarrow \R_{\geq 0}$ be symmetric submodular functions. For any parameter $\eps > 0$, there exists a collection of weights $\lambda_1, \dots , \lambda_m \in \R_{\geq 0}$ such that:
    \begin{enumerate}[(1)]
        \item $\lambda_1, \dots , \lambda_m$ constitutes a $(1 \pm \eps)$ sparsifier of $f_1, \dots , f_m$.
        \item At most $O(n / \eps^2)$ of the weights $\lambda_1, \dots , \lambda_m$ are non-zero. 
    \end{enumerate}
\end{theorem}

This improves upon the previous best sparsifier size from \cite{jambulapati2023sparsifying} of $O\left ( \frac{n \log(n/\eps) \cdot \log^{2.5}(n)}{\eps^2} \right )$ many terms. This also immediately yields sparsifiers of the same $O(n / \eps^2)$ many terms for sums of \emph{monotone}, submodular functions (where $f(S \cup \{v\}) \geq f(S)$ for every $S \subseteq [n]$, $v \in [n]$) via a reduction of \cite{khanna2024almost}.

\subsection{Technical Overview}

\paragraph{The Framework of \cite{RR26}}

In this section, we now present a brief technical overview for \cref{thm:main-intro}. As mentioned above, the proof of \cref{thm:main-intro} follows in the footsteps of \cite{RR26}. The main idea of \cite{RR26} is to apply an iterative ``weight perturbation'' procedure which slowly reduces the support size of an object, while simultaneously ensuring that the ``error accumulation'' is not too large.

In more detail, let $f_1, \dots , f_m: \R^n \rightarrow \R$ denote the collection of functions we wish to sparsify. As mentioned above, in our setting these will be the support functions of compact, centrally symmetric convex sets $C_1, \dots , C_m$.
To start, these functions $f_i$ are each given an implicit weight of $1$, so we let $w^{(1)}_i = 1$ for $i \in [m]$. The goal of \cite{RR26} is simple: they wish to find a ``weight perturbation'' vector $z^{(1)} \in [-1,1]^m$ such that the new weight vector $w^{(2)} = w^{(1)} + z^{(1)}$ satisfies:
\begin{enumerate}[(1)]
    \item $w^{(2)}$ is still a good sparsifier of $w^{(1)}$. I.e., for every $x \in \R^n$, 
    \[
    \sum_{i = 1}^m w^{(2)}_i \cdot f_i(x) \approx \sum_{i = 1}^m w^{(1)}_i \cdot f_i(x).
    \]
    \item $w^{(2)}$ has a smaller support than $w^{(1)}$. i.e., some entries of $w^{(2)}$ are now equal to $0$.
\end{enumerate}
Importantly, this second condition above exactly corresponds with finding a weight perturbation vector $z$ which has many coordinates \emph{equal} to $-1$. 

The key contribution from the work of \cite{RR26} is a general framework for establishing sufficient conditions for 
when such vectors $z$ exist. To do this, they define an appropriate ``one-sided sparsification polytope'':
\[
K_0 = \left \{ z \in [-1,1]^m: \sum_{i = 1}^m z_i \cdot f_i(x) \leq 0 \text{ for every }x \in \R^n \right \}.
\]
This is simply the set of all weight perturbation vectors which constitute good \emph{one-sided} sparsifiers, i.e., which yield $w^{(2)} = z + w^{(1)}$ such that for all $x \in \R^n$,
\[
\sum_{i = 1}^m w^{(2)}_i f_i(x) \leq \sum_{i = 1}^m w^{(1)}_i f_i(x).
\]
The key fact from \cite{RR26} is that if $K_0$ has large volume, then \emph{there must} exist weight perturbation vectors which satisfy both of the desired conditions above. 

\begin{claim}[Informal]\label{clm:informalVolumeTarget}
    Let $K_0$ be defined as above. If 
    \[
    \frac{\Vol_m(K_0)}{2^m} \geq 2^{-O(d)},
    \]
    then for any $\eps > 0$, there exists $(1 \pm \eps)$ sparsifiers of $f_1, \dots , f_m$ with total support $O(d / \eps^2)$.
\end{claim}

This clarifies the target statement that we now need in our setting: we define $f_1, \dots , f_m$ to be the support functions of centrally symmetric, compact, convex sets, and define $K_0$ to be one-sided sparsification polytope. Can we then lower bound the volume of $K_0$ by $2^m \cdot 2^{-O(n)}$?

\paragraph{Lower Bounding the Volume}

Fortunately, \cite{RR26} also presents tools for lower bounding the volume. Our starting point is their Lemma 18:

\begin{fact}[Lemma 18 of \cite{RR26}]
    Let $F:\R^m\to\R_{\ge0}$ be separately convex and not identically zero.  Suppose that, for some integer $d\ge0$,
\[
  F(\rho x)=\rho^dF(x)
  \qquad
  \text{for every }x\in\{-1,1\}^m\text{ and }\rho\ge0.
\]
Then
\[
  \Pr_{\sigma\in\{-1,1\}^m}[F(\sigma)>0]\ge e^{-2d}.
\]
\end{fact}

Of key importance is the notion of \emph{separate convexity} of the function $F$: for a convex set $\Omega \subseteq \R^n$, this essentially requires that for any $i\in[m]$ and any $\omega\in\Omega$, the restricted function $t\mapsto F(\omega + te_i)\in\R$ is convex, where $e_i\in\R^m$ is the $i^{\mathrm{th}}$ standard basis vector.

Nevertheless, this still raises new challenges; we must find a function $F$ which is separately convex (and satisfies the desired scaling property with $\rho$) such that being able to conclude $\Pr_{\sigma\in\{-1,1\}^m}[F(\sigma)>0]\ge e^{-2n}$ implies our desired volume lower bound on $K_0$. This is where we make our first key contribution; given our collection of compact, convex, centrally symmetric sets $C_1, \dots , C_m$, we assume that the total dimension spanned by the sum of $C_1, \dots , C_m$ is $n$ and define the function $F(t_1, \dots , t_m) = \Vol_n(\sum_{i = 1}^m t_i C_i)$. In this notation, $t_i \cdot C_i$ refers to scaling the set $C_i$ by $t_i$, and summing over bodies refers to the typical Minkowski sum of sets.

It is not hard to see that 
\[
 F(\rho x)=\rho^n F(x)
  \qquad
  \text{for every }x\in\{-1,1\}^m\text{ and }\rho\ge0;
\]
this follows from the fact that if one scales all coordinates by a factor of $\rho$ in a space of dimension $n$, then the resulting volume of the body shrinks by exactly a factor of $\rho^n$. 

The more challenging task is to show that this function $F$ is indeed separately convex. This relies crucially on a result in convex geometry related to Minkowski subtraction, which we first formally define. Given convex sets $A, B \subseteq \R^n$, we let their Minkowski difference be defined as
\[
A \ominus B := \bigcap_{y \in B} (A - y) = \{y \in \R^n: B +y \subseteq A\}.
\]
Note that $A \ominus B$ may potentially be empty, for instance in the case when $A \subsetneq B$. Nevertheless, for a proper convex body $A$ (i.e., $A$ has nonzero volume) and any (bounded) convex set $B$, there will always \emph{exist} some constant $r(A, B) >0$ such that $A \ominus r(A, B) \cdot B$ is non-empty.

More importantly, this now lets us build a parameterized function $g_{A, B}(r) = \mathrm{Vol}(A \ominus r \cdot B)$, for which we show the following fundamental property.
\begin{claim}\label{claim:g-convex}
For any convex bodies $A$ and $B$, the function $g_{A,B}$ is convex.
\end{claim}
The proof of \cref{claim:g-convex} follows from a formula for the derivative of $g_{A,B}(r)$ in terms of \emph{quermassintegrals} due to Bol~\cite{Bol43} and Matheron~\cite{Matheron78}. See \cref{sec:prelim} for details.
By carefully re-writing the function $F$ from above, we can directly show that the influence of each coordinate can be written in the form $g_{A, B}(r)$ (for carefully defined $A, B$), whereby we are able to conclude convexity in that coordinate, and thus all together, we obtain separate convexity. A formal proof that $F$ is separately convex is given in \cref{sec:proofOfVolume}.

Finally, it remains to show why a bound of the form $\Pr_{\sigma\in\{-1,1\}^m}[F(\sigma)>0]\ge e^{-2n}$ suffices for our target one-sided volume bound of \cref{clm:informalVolumeTarget}. Indeed, let us fix a vector $\sigma \in \{-1,1\}^m$ such that $F(\sigma)>0$. Letting $C_P = \sum_{i \in [m]: \sigma_i = 1} C_i$ and $C_N = \sum_{i \in [m]: \sigma_i = -1} C_i$, this implies that 
\[
F(\sigma) = \Vol(C_P \ominus C_N) >0,
\]
or equivalently, that $C_N \subseteq C_P$. By using known characterizations of convex bodies in terms of support functions, this then implies that for every $u \in \R^n$, 
\[
h_{C_N}(u) \leq h_{C_P}(u),
\]
or equivalently, that $h_{C_P}(u) - h_{C_N}(u)  \geq 0$. By again using properties of support functions, this is in turn equivalent to the fact that 
\[
\sum_{i = 1}^m \sigma_i \cdot h_{C_i}(u) \geq 0 \text{ for all } u \in \R^n.
\]
Thus, we know that 
\[
\Pr_{\delta \sim \{\pm1\}^m} \left [\sum_{i = 1}^m \sigma_i \cdot h_{C_i}(u) \geq 0 \text{ for all } u \in \R^n \right] \geq e^{-2n},
\]
and by simply replacing $\delta$ with $-\delta$, we then have that 
\[
\Pr_{\delta \sim \{\pm1\}^m} \left [\sum_{i = 1}^m \sigma_i \cdot h_{C_i}(u) \leq 0 \text{ for all } u \in \R^n \right] \geq e^{-2n}.
\]
A simple transformation of the convex bodies can then be used to show that the same statement holds when sampling from $[-1,1]^m$ rather than just $\{\pm1\}^m$, thereby yielding our desired volume bound (and sparsifiers). 

We remark that the recent work of~\cite{BKMT26} also proves a bound similar to~\cite[Lemma~18]{RR26} for the sparsification polytope arising in the context of sparsifying abelian Cayley graphs. Their paper uses very different techniques from the ones outlined here though, relying on group symmetries to prove such a bound. However, once the volume lower bound is established, the remaining machinery of~\cite{RR26} can be invoked as it is to obtain a sparsifier, as is also done in this paper. 

It would be interesting to find further such instances in sparsification where optimal sparsifiability follows volume lower bounds analogous to \cite[Lemma~18]{RR26}.

\subsection{Organization}

In \cref{sec:prelim} we present most of the necessary convex geometric preliminaries. In \cref{sec:proofOfVolume}, we present a formal proof of the necessary one-sided volume bound, which we then use in \cref{sec:sparsifiersConvex} to build sparsifiers for sums of support functions of convex bodies. In \cref{sec:sparsifiersSemiNorm}, we show this implies optimal size sparsification of sums of seminorms, and in \cref{sec:applications}, we further use these seminorm sparsifiers to build optimal sparsifiers of hypergraph cut functions and sums of symmetric submodular functions. 

\subsection{Acknowledgments}

The main result of this paper was discovered through a sequence of discussions with ChatGPT 5.5-Pro and 5.6-Sol. This initially began with an exploration of hypergraph cut sparsification, which the authors later realized could be extended to sparsification of seminorms, and ultimately its final form of ``Minkowski sum sparsification'' in conversation with ChatGPT. These initial results had size complexity $O \left (\frac{n \log(1 / \eps)}{\eps^2} \right ) $, which the authors then simplified and improved via careful use of the framework of \cite{RR26}, again in conversation with ChatGPT. All of the text in this paper is written by the authors.

A. B. thanks Pravesh Kothari for insightful conversations related to sparsification. J. B. thanks Venkatesan Guruswami for valuable discussions and encouragement. A. P. thanks Sanjeev Khanna and Madhu Sudan for useful conversations related to sparsification.

\section{Convex Geometry Preliminaries}\label{sec:prelim}

We now formally define a number of concepts in convex geometry which we need for our study. We closely follow the presentations of Gruber~\cite{Gruber07} and Schneider~\cite{Schneider14}. Given a positive integer $n \in \mathbb N$, we say that a set $A \subseteq \R^n$ is \emph{convex} if for all $x, y \in A$ and all scalars $\lambda \in [0, 1]$, we have that $\lambda x + (1-\lambda)y \in A$. If in addition $A$ is also non-empty and compact (i.e., closed with finite volume), then we say that $A$ is a \emph{convex body}. We let $\cK^n$ denote the set of all such convex bodies $A$ in $\R^n$.

We say that a convex body $A$ is \emph{proper} if $A$ has non-zero interior. That is, there is a point $x \in A$ and a radius $r > 0$ such that the Euclidean ball $B_n(x, r) \subseteq A$. We let $\cK^n_n \subsetneq \cK^n$ denote the set of proper convex bodies in $\R^n$.

We say that $A \in \cK^n$ is \emph{symmetric} (a.k.a. \emph{centrally symmetric}) if for all $x \in \R^n$, we have that $x \in A$ if and only if $-x \in A$. We let $\cK^n_{\sym}$ denote the set of symmetric convex bodies, and $\cK^n_{n, \sym}$ denote the set of proper, symmetric convex bodies.

\subsection{Minkowski Sums and Differences} Convex bodies support a number of algebraic operations. Given a scalar $\lambda \in \mathbb R$ and a convex body $A \in \cK^n$, we let $\lambda A := \{\lambda x \mid x \in A\}$. Furthermore, given a vector $y \in \mathbb R^n$, we denote the \emph{translate} $A + y := \{x + y \mid x \in A\}$, with $A - y := A + (-y)$.

 For any $A, B \in \cK^n$ we can define their \emph{Minkowski sum} to be
\[
    A + B := \bigcup_{y \in B} (A + y) = \{x + y \mid x \in A, y \in B\} \in \cK^n.
\]
A more subtle notion one can also define is the \emph{Minkowski difference}, which is defined as (see e.g., \cite[p.146]{Schneider14}\footnote{Schneider~\cite{Schneider14} uses the notation $A \div B$ for Minkowski difference.}) 
\begin{align}
A \ominus B := \bigcap_{y \in B} (A - y) = \{y \in \mathbb R^n \mid B + y \subseteq A\}.\label{eq:minkowski-diff}
\end{align}
It is easy to see that $A \ominus B$ is always convex, as it is defined to be an intersection of convex bodies. However, $A \ominus B$ may be empty and thus disqualify it from being a convex body. It is easy to observe that $A \ominus B$ is nonempty if and only if there exists a translate of $B$ contained in $A$. In the special case in which $A$ and $B$ are both centrally symmetric, we can see that $A \ominus B$ (if nonempty) is centrally symmetric as well, as $B + y \subseteq A$ if and only if $-(B+y) = B - y \subseteq -A = A$. In particular, this means that $0 \in A \ominus B$, so $B \subseteq A$ by \cref{eq:minkowski-diff}. As such, we have the following fact (observed in, e.g., \cite{RR26}).

\begin{fact}\label{fact:central-contained}
Let $A, B \in \cK^n_{\sym}$ have the property that $A \ominus B \neq \emptyset$, then $B \subseteq A$.
\end{fact}

We now state some standard facts about Minkowski addition and subtraction (e.g., Chapter~3 of \cite{Schneider14} or Proposition~11 in \cite{RR26}).

\begin{fact}\label{fact:basics}
For any $A, B, C \in \cK^n$ and scalars $\alpha, \beta \ge 0$, we have that.
\begin{enumerate}[(1)]
\item $(A + B) + C = A + (B + C)$.
\item $(A \ominus B) \ominus C = A \ominus (B + C)$.
\item $(A + B) \ominus B = A$.
\item $\alpha A + \beta A = (\alpha + \beta) A$
\item $\alpha A + \alpha B = \alpha (A + B)$.
\item $\alpha A \ominus \beta A = (\alpha - \beta) A$ if $\alpha \ge \beta$.
\end{enumerate}
\end{fact}

As previously discussed, for arbitrary $A, B \in \cK^n$, $A \ominus B$ may be empty. However, if we scale $B$ by some factor $\lambda B$, then the \emph{parallel body} $A \ominus \lambda B$ is non-empty for $\lambda \ge 0$ sufficiently small. As such, we can define the \emph{inradius} of $A$ relative to $B$ as follows \cite[p.148]{Schneider14}.
\begin{align}
    r(A; B) := \max \{\lambda \ge 0 \mid \exists x \in \R^n, \lambda B + x \subseteq A\}.\label{eq:inrad}
\end{align}
Note that the maximum here is well-defined as $A$ and $B$ are compact.

\subsection{Volumes and Mixed Volumes} Given a convex body $A \in \cK^n$, we let $\Vol_n(A)$ denote the standard Euclidean volume of $A$. Note that $\Vol_n(A) = 0$ if and only if $A \in \cK^n \setminus \cK^n_n$. By convention, we also let $\Vol_n(\emptyset) = 0$. For some of our computations, we also need a notion of \emph{mixed volume}. In particular, for any $A_1, \hdots, A_n \in \cK^n$, we define their mixed volume to be (see \cite[Proposition 6.7]{Gruber07})
\[
    \Vol_n(A_1, \hdots, A_n) := \frac{1}{d!}\left[\sum_{S \subseteq [n]} (-1)^{n-|S|}\Vol_n\left(\sum_{i \in S}A_i\right)\right].
\]
The key defining property of mixed volumes is that for any $m \in \N$ and $A_1, \hdots, A_m \in \cK^n_n$ and scalars $\lambda_1, \hdots, \lambda_m \ge 0$, we have that (\cite[Theorem 6.5]{Gruber07})
\[
    \Vol_n(\lambda_1 A_1 + \cdots + \lambda_m A_m) := \sum_{(i_1, \hdots, i_n) \in [m]^n} \lambda_{i_1} \cdots \lambda_{i_n} \Vol_n(A_{i_1}, \hdots, A_{i_d}).
\]
As such, we have that $\Vol_n(A, \hdots, A) = \Vol_n(A)$ for any $A \in \cK^n$. Further observe that both volume and mixed volume are translation invariant~\cite[Proposition 6.5]{Gruber07}. That is, $\Vol_n(A_1, \hdots, A_n) = \Vol_n(A_1 + x_1, \cdots, A_n + x_n)$ for all $x_1, \hdots, x_n \in \R^n$.

\begin{fact}[Theorem 6.9 and Corollary 6.1~\cite{Gruber07}]\label{fact:mixed}
Mixed volume is always nonnegative and is monotone in all of its arguments. That is, for any $A_1, \hdots, A_n, B_1, \hdots, B_n \in \cK^n$ with $A_i \subseteq B_i$ for all $i \in [n]$, we have that
\[
    0 \le V(A_1, \hdots, A_n) \le V(B_1, \hdots, B_n).
\]
\end{fact}

\subsection{Quermassintegrals and Differentiability}\label{sec:differentiability}

In this work, we are only concerned with mixed volumes $\Vol_n(A_1, \hdots, A_n)$ where the set $\{A_1, \hdots, A_n\}$ has at most two elements. This special case has the name \emph{quermassintegrals} and has the following special notation. For convex bodies $A, B \in \cK^n$ and integer $i \in \{0, 1, \hdots, n\}$, we denote the $i$th quermassintegral as
\[
    W_i(A; B) := \Vol_n(\underbrace{A, \hdots, A}_{n-i}, \underbrace{B, \hdots, B}_i).
\]
In particular, $W_0(A; B) = \Vol_n(A)$ and $W_n(A; B) = \Vol_n(B)$.

A key application of these quermassintegrals is differentiating volumes of Minkowski differences. More precisely, given convex bodies with $A, B \in \cK^n$, we define for all $\lambda \ge 0$ the quantity
\[
    f_{A,B}(\lambda) :=
    \Vol_n(A \ominus \lambda B).
\]
We observe that this function is continuous at $\lambda = r(A;B)$.
\begin{proposition}
$\Vol_n(A \ominus r(A;B) B) = 0$.
\end{proposition}
\begin{proof}
Assume for sake of contradiction that $\Vol_n(A \ominus r(A;B) B) > 0$. Thus, there exists a Euclidean ball $B_n(x, \eps) \subseteq A \ominus r(A;B) B$. Since $B$ is compact, there exists some $\lambda > 0$ such that $\lambda B + x \subseteq B_n(x, \eps)$. That is, $\lambda B + x \subseteq A \ominus r(A;B) B$. By \cref{fact:basics}, this is equivalent to $x \in A \ominus (r(A;B) + \lambda) B$. By \cref{eq:inrad}, this implies that $r(A;B) \ge r(A;B) + \lambda$, a contradiction.
\end{proof}

Rather importantly, this function is (usually) differentiable.

\begin{fact}[\cite{Bol43,Matheron78}, adapted from Proposition~2.6 of \cite{RSG22}]\label{fact:f-dir}
Assume $A \in \cK^n$ and $B \in \cK^n_n$, then $f_{A,B}$ is differentiable in the range $(0, r(A; B))$ with
\begin{align}
    f'_{A,B}(\lambda) := -n W_1(A \ominus \lambda B; B).\label{eq:f-dir}
\end{align}
Furthermore, (\ref{eq:f-dir}) gives the right-derivative of $f_{A,B}(\lambda)$ at $\lambda = 0$ and the left-derivative at $\lambda = r(A;B)$.
\end{fact}

For further context and generalizations of \cref{fact:f-dir}, see some papers by Hern\'{a}ndez Cifre and Saor\'{\i}n~\cite{CS10diff,CS10vol}.

Extending \cref{fact:f-dir}, we show that $f_{A,B}$ is convex on the entire interval $[0, \infty)$, even when $B$ is improper.

\begin{lemma}\label{lemma:f-convex}
For any $A, B \in \cK^n$, we have that $f_{A, B}$ is convex on $[0, \infty)$.
\end{lemma}
\begin{proof}
First assume $B$ is proper (i.e., $B \in \cK^n_n$), then by \cref{fact:f-dir}, we know that the derivative $f'_{A,B}(\lambda) = -n W_1(A \ominus \lambda B; B)$ for all $\lambda \in [0, r(A;B)]$. Note that this derivative is negative and decreasing by \cref{fact:mixed}. Thus, $f_{A,B}$ is convex, non-negative, and decreasing on $[0, r(A;B)]$. Since $f_{A,B}(\lambda) = 0$ for all $\lambda \ge r(A;B)$, we can immediately deduce that $f_{A,B}$ is convex on all of $[0, \infty)$. Alternatively, notice that the left derivative of $f_{A,B}$ at $r(A;B)$ is $-n W_1(A \ominus r(A;B) B; B) \le 0$, while the right derivative is $0$, so convexity is preserved along the entirety of $f_{A,B}$.

If $B \in \cK^n \setminus \cK^n_n$, for any $\eps > 0$, consider $B_{\eps} := B + B_n(0, \eps)$. Then, $f_{A, B_{\eps}}$ is convex on $[0, \infty)$. Thus, it suffices then to show for all $\lambda \in [0, \infty)$, we have pointwise convergence to $f_{A,B}$ as $\eps \to 0$. More formally, we claim that
\begin{align}
    \lim_{\eps \to 0^{+}} f_{A, B_{\eps}}(\lambda) = f_{A, B}(\lambda).\label{eq:f-converge}
\end{align}
This fact is trivial for $\lambda = 0$, as $f_{A,B_{\eps}}(0) = \Vol_n(A)$, so assume $\lambda > 0$. Observe that we always have that since $B \subseteq B_{\eps}$ for any $\eps > 0$, we have that $A \ominus \lambda B_\eps \subseteq A \ominus \lambda B$. If $\Vol_n(A \ominus \lambda B) = 0$ or $\lambda > r(A;B)$, then convergence is guaranteed by monotonicity and non-negativity.  Thus, assume that $\Vol_n(A \ominus \lambda B) > 0$. Using a second invocation of \cref{fact:f-dir}, we have that $\Vol_n((A \ominus \lambda B) \ominus B_n(0, \delta))$ is a continuous function of $\delta$. Thus, by \cref{fact:basics}, we have that
\[
    \lim_{\delta \to 0^{+}} \Vol_n(A \ominus \lambda B_{\delta/\lambda}) = \lim_{\delta \to 0^{+}} \Vol_n((A \ominus \lambda B) \ominus B_n(0, \delta)) = \Vol_n(A \ominus \lambda B).
\]
Thus, setting $\eps = \delta / \lambda$, we get \cref{eq:f-converge}. Thus, $f_{A,B}$ is convex.
\end{proof}

\subsection{Support Functions}

Our sparsification theorem will crucially rely on \emph{support functions} of convex bodies. For a nonempty compact convex set $K\subseteq \R^n$, its support function is
\[
  h_K(u)=\max_{x\in K}\ip{x, u},
  \qquad u\in \R^n.
\]
We have the following fact:
\begin{fact}[e.g., Chapter 13 of \cite{Rockafellar70} or \cite{DK00}]\label{fact:supportProperties}
Support functions satisfy
\[
  h_{K+L}=h_K+h_L,
  \qquad
  h_{tK}=t h_K \quad (t\ge 0),
\]
and
\[
  K\subseteq L
  \quad\Longleftrightarrow\quad
  h_K(u)\le h_L(u)\text{ for every }u\in \R^n.
\]
\end{fact}

Likewise, we also have the following relationship between support functions and Minkowski subtraction:

\begin{fact}[e.g., Proposition 2 of \cite{DK00}]\label{fact:supportFnCharacterization}
For compact convex sets $A, B$,
\[
z\in A\ominus B
    \iff \ip{z,u}\leq h_A(u)-h_B(u)\text{ for every }u\in \R^n.
\]
\end{fact}

\section{A One-Sided Volume Lemma}\label{sec:proofOfVolume}

With the above preliminaries established, we now proceed to the key technical lemma which underlies our sparsification theorem. Just as in \cite{RR26} and \cite{BKMT26}, this relies on a ``one-sided volume lemma''. In our case this one-sided volume lemma will correspond with appropriate one-sided sparsifiers of \emph{sums of support functions of convex bodies}:

\begin{lemma}\label{lem:oneSidedVolumeBound}
    Let $C_1,\ldots,C_m$ be compact centrally symmetric convex sets whose Minkowski sum has dimension $d$ in $\R^n$. Define
\[
    K_0:=\left \{z\in[-1,1]^m:
       \sum_{i=1}^m z_i h_{C_i}(u)\leq0
       \text{ for every }u \right \}.
\]
Then
\[
    \frac{\mathrm{vol}_m(K_0)}{2^m}\geq e^{-2d}.
\]
\end{lemma}

\subsection{The Framework of \cite{RR26}}
Let $\Omega\seq\R^m$ be a convex set. A function $F:\Omega\to\R$ is said to be \emph{separately convex} if for any $i\in[m]$ and any $\omega\in\Omega$, the restricted function $t\mapsto F(\omega + te_i)\in\R$ is convex, where $e_i\in\R^m$ is the $i^{\mathrm{th}}$ standard basis vector. Separately convex functions satisfy the following convenient multi-variate generalization of Jensen's theorem \cite[Lemma~12]{RR26}: Let $X_1, \ldots, X_m$ be independent $\R$-valued random variables such that $(X_1, \ldots, X_m)$ is supported on $\Omega$. Then for any separately convex function $F:\R^m\to\R$ we have $\E F(X_1, \ldots, X_m)\geq F(\E X_1, \E X_2, \ldots, \E X_m)$.

Importantly, the work of \cite{RR26} provides a powerful framework for bounding one-sided volume by using the notion of separate convexity:

\begin{fact}[Lemma 18 of \cite{RR26}]\label{fact:RR26VolumeBound}
    Let $F:\R^m\to\R_{\ge0}$ be separately convex and not identically zero.  Suppose that, for some integer $d\ge0$,
\[
  F(\rho x)=\rho^dF(x)
  \qquad
  \text{for every }x\in\{-1,1\}^m\text{ and }\rho\ge0.
\]
Then
\[
  \Pr_{\sigma\in\{-1,1\}^m}[F(\sigma)>0]\ge e^{-2d}.
\]
\end{fact}

Note that it is not a priori obvious exactly what the function $F$ should be. In what follows, we show that for an appropriate choice of this function, we can simultaneously (a) satisfy the conditions of the preceding fact, and (b) take advantage of the conclusion of the above fact to conclude our desired volume lemma.

\subsection{The Proof of \cref{lem:oneSidedVolumeBound}}

To begin, we will define a function which is separately convex and not identically $0$, thereby satisfying the conditions of \cref{fact:RR26VolumeBound}. 

To start, as in the statement of \cref{lem:oneSidedVolumeBound}, we let $C_1, \dots , C_m$ be centrally symmetric convex bodies which span a $d$ dimensional subspace of $\R^n$. For $t = (t_1, \dots , t_m) \in \R^m$, we let 
\[
\mathcal{Z}(t) = \left \{ z \in \R^n: \langle z, u \rangle \leq  \sum_{i = 1}^m t_i h_{C_i}(u) \text{ for every } u \in \R^n \right \}
\]
Importantly, we have the following facts about $\mathcal{Z}(t)$:

\begin{claim}\label{clm:usefulZProperties}
    $\mathcal{Z}(t)$ is closed and convex.  It is also bounded whenever it is nonempty. Moreover, for $a\in\R^m_{\ge 0}$, 
\[
  \mathcal{Z}(a)=\sum_{i=1}^m a_i C_i.
\]
\end{claim}

\begin{proof}
    To see the final item, observe that for any $a \in \R_{\geq 0}^m$, the right-hand side of the definition of $\mathcal{Z}(a)$ is the support function of $\sum_i a_i C_i$. It then immediately follows that $\mathcal{Z}(a)=\sum_{i=1}^m a_i C_i.$
\end{proof}

Now, we define the function \[
\Phi(t)=\Vol_d(\mathcal{Z}(t)),
\]
with the understanding that $\Phi(t)=0$ if $\mathcal{Z}(t)=\varnothing$. Crucially, $\Phi$ is separately convex:

\begin{claim}\label{clm:separatelyConvex}
    For $C_1, \dots , C_m$ and $\Phi: \R^m \rightarrow \R_{\geq 0}$ as defined above, $\Phi$ is separately convex. 
\end{claim}

\begin{proof}
To prove this, we rely crucially on \cref{lemma:f-convex}. We fix any $R > 0$, we set $c = R \cdot \mathbf{1} \in \R^m$, and take $t \in [-R, R]^m$. By \cref{fact:supportFnCharacterization}, we then see that 
\[
\mathcal Z(t)=\mathcal Z(c) \ominus \mathcal Z(c-t).
\]
Now, we fix all coordinates except $t_j$, for any choice of $j \in [m]$. We write 
\[
  D=\sum_{i\ne j}(R-t_i)C_i,
  \qquad
  A=\mathcal{Z}(c)\ominus D.
\]

By \cref{fact:basics},
\[
  \mathcal{Z}(t)
  =A\ominus (R-t_j)C_j.
\]

 If $A=\varnothing$, then this section of $\Phi$ is identically zero.  Otherwise $A$ is a nonempty compact convex set and the map
\[
  t_j\longmapsto\Phi(t)
  =\mathrm{Vol}_d\bigl(A\ominus (R-t_j)C_j\bigr)
\]
is convex by  \cref{lemma:f-convex}. Thus $\Phi$ is separately convex on $[-R,R]^m$.  Since $R$ is arbitrary, it is separately convex on all of $\R^m$.
\end{proof}

We can then proceed to a proof of \cref{lem:oneSidedVolumeBound}:

\begin{proof}[Proof of \cref{lem:oneSidedVolumeBound}]
    We let $E$ denote the span of the bodies $C_1, \dots , C_m$, and let $d$ denote the dimension of $E$. If $d = 0$, the claim is immediate. Otherwise, we define $\Phi$ as above. Note that for any $x \in \{\pm 1\}^m$ and any $\rho \geq 0$, we have that 
    \[
    \Phi(\rho x) = \mathrm{Vol}_d(\mathcal{Z}(\rho x)) = \mathrm{Vol}_d\left (\sum_{i = 1}^m \rho \cdot x_i \cdot C_i \right ) = \mathrm{Vol}_d\left (\rho \cdot \sum_{i = 1}^m  x_i \cdot C_i \right ) 
    \]
    \[
    = \rho^d \cdot \mathrm{Vol}_d\left ( \sum_{i = 1}^m  x_i \cdot C_i \right ) = \rho^d \cdot \Phi( x). 
    \]
    
    Now, in conjunction with \cref{clm:separatelyConvex}, this then implies that 
    \[
    \Pr_{\delta \sim \{\pm 1\}^m}[\Phi(\delta) > 0] \geq e^{-2d}.
    \]

    Now, consider any vector $\delta \in \{\pm1\}^m$ such that $\Phi(\delta) > 0$. Let $P = \{i \in [m]: \delta_i = 1\}$, and let $N = [m] - P$. We let $C_P$ denote the sum of the convex bodies corresponding to $P$; i.e., $C_P = \sum_{i \in P} C_i$, and we let $C_{N} = \sum_{i \in N} C_i$. By \cref{clm:usefulZProperties}, we then see that $\mathcal{Z}(\delta) = C_P \ominus C_N$. Because $\Phi(\delta) > 0$, this then also means that $\mathcal{Z}(\delta)$ is non-empty. Because $C_P$ and $C_N$ are centrally symmetric, 
    it is then the case (by \cref{fact:central-contained}) that $C_N \subseteq C_P$.
    In the view of support functions, this then means that for every $u \in E$,
    \[
    \sum_{i \in N}  h_{C_i}(u) \leq  \sum_{i \in P}  h_{C_i}(u),
    \]
    or equivalently, that for every $u \in E$,
    \[
    \sum_{i = 1}^m \delta_i h_{C_i}(u) \geq 0
    \]
    This implies that 
    \[
    \Pr_{\delta \sim \{\pm 1\}^m}[\delta_i h_{C_i}(u) \geq 0 \text{ for all } u \in E] \geq e^{-2d};
    \]
    by considering $-\delta$ instead of $\delta$, this then also implies that 
    \[
     \Pr_{\delta \sim \{\pm 1\}^m}[\delta_i h_{C_i}(u) \leq 0 \text{ for all } u \in E] \geq e^{-2d}.
    \]

    Finally, it remains to translate this result from the setting of $\{\pm1\}^m$ vectors to $[-1,1]^m$ vectors. For this, we sample a random vector $Z \in [-1,1]^m$. We write the vector $Z$ as $Z_i = R_i \cdot \eps_i$, where $R_i = |Z_i|$ and $\eps_i = \mathrm{sgn}(Z_i)$. The variables $R_i$ are uniform on $[0,1]$ and $\eps_i$ is uniform from $\{\pm1\}$. Now, we can simply invoke the same reasoning from above on the collection of convex bodies $R_1 \cdot C_1, \dots , R_m \cdot C_m$. On these bodies, we then see that for every choice of $R_i$,
    \[
     \Pr_{\eps \sim \{\pm 1\}^m}[\eps_i \cdot R_i h_{C_i}(u) \leq 0 \text{ for all } u \in E] \geq e^{-2d}.
    \]
    Thus, 
    \[
    \Pr_{Z \sim [-1,1]^m }[Z_i \cdot  h_{C_i}(u) \leq 0 \text{ for all } u \in E] 
    \]
    \[
    = \Pr_{R \sim [0,1]^m}\Pr_{\eps \sim \{\pm 1\}^m}[\eps_i \cdot R_i h_{C_i}(u) \leq 0 \text{ for all } u \in E] \geq e^{-2d}.
    \]
    This concludes the claim. 
\end{proof}

\section{Building Support Sparsifiers for Sums of Convex Bodies}\label{sec:sparsifiersConvex}

With \cref{lem:oneSidedVolumeBound} established, we now invoke the framework of \cite{RR26} for building linear-size sparsifiers. Note that the material of this section is almost entirely from the work of \cite{RR26}.

\subsection{\cite{RR26} Preliminaries}

We start by recording a few facts from \cite{RR26}. To do this, we first require the notion of the restriction of a convex-set: for $K \subseteq [-1,1]^m$, we let $K_S = \{ x' \in \R^S: (x', \mathbf{0}^{[m] - S}) \in K\}$.

\begin{fact}[Reis--Rothvoss symmetrization, Theorem 16 from \cite{RR26}]\label{thm:RR-sym}
Let $p,\eta\in(0,1/2]$, and let $K\subseteq[-1,1]^m$ be convex.  Suppose that $[-\eta,\eta]^m\subseteq K$ 
and
\[
  \frac{\Vol_S(K_S)}{2^{|S|}}\ge p
  \qquad\text{for every }S\subseteq[m].
\]
If
\[
  m\ge\frac{\log_2(1/p)}{\eta^2},
\]
then
\[
  \Vol_m(K\cap(-K))\ge2^{-5m}.
\]
\end{fact}

We also use the following fact which shows that (appropriate scalings of) large enough symmetric convex bodies contain vectors with many integral coordinates:

\begin{fact}[Large-body partial coloring, Theorem 7 from \cite{RR26}]\label{thm:RR-partial}
For every constant $c>0$, there is a constant $s=s(c)>0$ such that the following holds.  If $Q\subseteq[-1,1]^m$ is a symmetric convex body with
\[
  \Vol_m(Q)\ge c^m,
\]
then there is a vector
\[
  z\in sQ\cap[-1,1]^m
\]
for which at least $m/2$ coordinates satisfy $|z_i|=1$.
\end{fact}

\subsection{The Partial Coloring}

In this section, we now prove the following lemma:

\begin{lemma}\label{lem:body-partial}
There are universal constants $A,\kappa>0$ with the following property.  Let $C_1,\dots,C_m$ be compact centrally symmetric convex sets contained in a $d$-dimensional space, and put $C=\sum_iC_i$.  If $m\geq \kappa d$, then there is $z\in[-1,1]^m$ such that
\begin{enumerate}[(1)]
  \item at least $m/4$ coordinates of $z$ equal $-1$; and 
  \item for every $u$,
  \[
    \left | {\sum_{i=1}^m z_i h_{C_i}(u)} \right |
    \le A\sqrt{\frac dm}\,h_C(u).
  \]
\end{enumerate}
\end{lemma}

\begin{proof}
    The proof relies on both \cref{thm:RR-sym} and \cref{thm:RR-partial} from above. To start, we can observe that when $d = 0$, the claim is trivial. Thus, we focus only on the case when $d \geq 1$. We set 
    \[
  p_0=e^{-2d},
  \qquad
  \beta=2\log_2 e,
  \qquad
  \eta=\sqrt{\frac{\beta d}{m}}.
\]

Next, choose $\kappa \geq 4 \beta$; in this way if $m \geq \kappa \cdot d$, it must be that $\eta \leq 1/2$. Now, we define a new convex set $K_{\eta}$:
\[
K_\eta
  =\left \{x\in[-1,1]^m:
  \sum_{i=1}^m x_i h_{C_i}(u)
  \leq \eta h_C(u)
  \text{ for every }u \right \}.
\]
This set is a relaxation of the body considered in \cref{lem:oneSidedVolumeBound} as we now allow for some slack on the right-hand side. Importantly, this set now ensures that $[-\eta, \eta]^m \subseteq K_{\eta}$. Indeed, for any $x \in[-\eta, \eta]^m$, 
\[
  \sum_{i=1}^m x_i \cdot h_{C_i}(u)
  \le \sum_{i=1}^m|x_i| \cdot h_{C_i}(u)
  \le \eta \cdot h_C(u).
\]

Now, fix $S \subseteq [m]$. Because $(K_0)|_S \subseteq (K_{\eta})|_S$, we can invoke \cref{lem:oneSidedVolumeBound} to see that 
\[
\frac{\mathrm{Vol}_{|S|}((K_{\eta})|_S)}{2^{|S|}} \geq e^{-2d}.
\]

Furthermore,
\[
  \log_2(1/p_0)=2d\log_2 e=\beta d
  \quad\text{and}\quad
  m=\frac{\log_2(1/p_0)}{\eta^2}.
\]
Thus \cref{thm:RR-sym} applies and gives
\[
  \Vol_m(Q_\eta)\ge 2^{-5m},
\]
where
\[
  Q_\eta
  =K_\eta\cap(-K_\eta)
  =\left \{x\in[-1,1]^m:
  \left | \sum_{i=1}^m x_i h_{C_i}(u) \right |
  \le \eta h_C(u)
  \text{ for every }u\right \}.
\]
Next, we invoke \cref{thm:RR-partial} with $c=2^{-5}$.
This implies that there is a universal constant $s>0$ and a vector
\[
  z\in s \cdot Q_\eta\cap[-1,1]^m
\]
with at least $m/2$ saturated coordinates (i.e., such that $z_i \in \{\pm 1\}$). Moreover, by taking $-z$ if necessary, we additionally are guaranteed that at least $m/4$ of these coordinates equal $-1$. Because $\frac{z}{s} \in Q_{\eta}$, this implies that for every $u$,
\[
  \left | \sum_{i=1}^m z_i h_{C_i}(u) \right |
  \leq s \cdot \eta \cdot h_C(u)
  =s\cdot \sqrt{2\log_2 e} \cdot \sqrt{\frac{d}{m}} \cdot h_C(u).
\]
This implies the claimed result with $A=s\sqrt{2\log_2 e}$.
\end{proof}

\subsection{Building the Sparsifier}

We now prove our desired sparsification statement:

\begin{theorem}[Minkowski-sum sparsification]\label{thm:body-main}
There is a universal constant $C_{\star}>0$ with the following property.  Let $E$ be a finite-dimensional real vector space, and let $C_1,\dots,C_m\subseteq E$ be compact centrally symmetric convex sets.  Put
\[
  C=C_1+\cdots+C_m,
  \qquad
  d=\dim\mathrm{span}(C).
\]
For every $0<\eps<1$, there exist coefficients $\lambda_1,\dots,\lambda_m\ge 0$, supported on at most $C_{\star} \cdot d/\eps^2$ indices, such that for every $u \in E$,
\[
  (1-\eps) \cdot \sum_{i = 1}^m h_{C_i}(u)
  \leq\sum_{i = 1}^m \lambda_i \cdot h_{C_i}(u)
  \leq (1+\eps) \cdot \sum_{i = 1}^m h_{C_i}(u).
\]
\end{theorem}

\begin{proof}
Our proof closely follows the strategy used to prove \cite[Theorem 26]{RR26}. To start, we discard all indices for which $h_{C_i}(u)$ is identically zero.
Note that if $d=0$, then $\sum_{i = 1}^m h_{C_i}(u)\equiv0$, and choosing every $\lambda_i=0$ proves the theorem.  Hence assume $d\geq1$.
Note also that it is enough to consider $0<\eps\le 1/2$.  Indeed, for $1/2<\eps<1$, applying the result with accuracy $1/2$ gives a stronger approximation and uses $O(d)=O(d/\eps^2)$ summands.

Let $A$ and $\kappa$ be the constants from \cref{lem:body-partial}, and set
\[
  q=\sqrt{\frac34},
  \qquad
  B=\frac{A}{1-q}.
\]
Next, we choose a universal constant $C_\star$ so large that
\begin{equation}\label{eq:Cstar-choice}
  \tau_\star\ge 4\kappa,
  \qquad
  \frac{B}{\sqrt{C_\star}}\le\frac12.
\end{equation}
Define the target support size
\begin{equation}\label{eq:M-def}
  M=\left\lceil C_\star \cdot \frac{d}{\eps^2}\right\rceil.
\end{equation}

The guiding intuition for building the sparsifier is that \cref{lem:body-partial} provides a ``weight perturbation'' to be applied to the sparsifier (whose weights start off as being all $1$'s). In this way, when a support functions receive a weight perturbation of $-1$, this zeroes out the contribution from that support function. We then iteratively apply this procedure until we reach our desired sparsity. 

We maintain nonnegative weights $w_i^{(t)}$ and the current Minkowski sum
\[
  C^{(t)}=\sum_{i=1}^m w_i^{(t)}C_i.
\]
Initially $w_i^{(0)}=1$, so $C^{(0)}=C$.  Let
\[
  S_t=\{i:w_i^{(t)}>0\},
  \qquad
  m_t=|S_t|.
\]
If $m_t\le M$, we terminate the algorithm, as this implies we have reached our desired sparsity.  Otherwise $m_t>M\ge\kappa d$, and we apply \cref{lem:body-partial} to the active summands
\[
  D_i^{(t)}=w_i^{(t)}C_i,
  \qquad i\in S_t.
\]
By doing this, we obtain $z^{(t)}\in[-1,1]^{S_t}$ with at least $m_t/2$ saturated coordinates such that for every $u$:
\begin{equation}\label{eq:step-disc}
  \left | {\sum_{i\in S_t}z_i^{(t)}h_{D_i^{(t)}}(u)} \right |
  \le \delta_t h_{C^{(t)}}(u),
  \qquad
  \delta_t=A\sqrt{\frac{d}{m_t}}.
\end{equation}
By $m_t>M$ and \eqref{eq:Cstar-choice},
\[
  \delta_t < \frac{A}{\sqrt{C_\star}}\eps
  \le \frac{B}{\sqrt{C_\star}}\eps
  \le \frac{\eps}{2}<1.
\]
As proved in \cref{lem:body-partial}, at least $m_t/4$ of the saturated coordinates are equal to $-1$, i.e.,
\[
  \left | \left \{i\in S_t:z_i^{(t)}=-1 \right \} \right |\ge\frac{m_t}{4}.
\]
We now apply this $z$ as our weight perturbation: we update
\begin{equation}\label{eq:weight-update}
  w_i^{(t+1)}=(1+z_i^{(t)})w_i^{(t)}
  \qquad(i\in S_t),
\end{equation}
and retain weight zero outside $S_t$.  The weights remain nonnegative, and every coordinate with $z_i^{(t)}=-1$ disappears.  Consequently,
\begin{equation}\label{eq:support-shrink}
  m_{t+1}\le\frac34m_t.
\end{equation}

The second detail is that we must manage the deterioration of the sparsifier accuracy. For this, we can observe that 
\begin{align*}
  h_{C^{(t+1)}}(u)
  &=\sum_{i\in S_t}(1+z_i^{(t)})h_{D_i^{(t)}}(u)\\
  &=h_{C^{(t)}}(u)
    +\sum_{i\in S_t}z_i^{(t)}h_{D_i^{(t)}}(u).
\end{align*}
Thus \eqref{eq:step-disc} and the support-function characterization of inclusion imply that for every $u$,
\begin{equation}\label{eq:step-inclusion}
  (1-\delta_t) \cdot h_{C^{(t+1)}}(u)
  \leq h_{C^{(t)}}(u)
  \leq (1+\delta_t) \cdot h_{C^{(t+1)}}(u).
\end{equation}

The preceding equation bounds the accuracy deterioration in a single round of weight updates. We now bound the total result of these updates across all iterations. We let $T$ be the first index for which $m_T\le M$.  If $T=0$, there is nothing to prove.  Otherwise $m_{T-1}>M$.  From \eqref{eq:support-shrink}, for $0\le t\le T-1$,
\[
  m_t
  \ge \left(\frac43\right)^{T-1-t}m_{T-1}.
\]
Therefore
\begin{align}
  \sum_{t=0}^{T-1}\delta_t
  &\le A\sqrt{\frac{d}{m_{T-1}}}
  \sum_{j=0}^{\infty}\left(\sqrt{\frac34}\right)^j\notag\\
  &<\frac{A}{1-\sqrt{3/4}}\sqrt{\frac{d}{M}}
  =B\sqrt{\frac{d}{M}}
  \le \frac{B}{\sqrt{C_\star}}\eps
  \le\frac{\eps}{2}.
  \label{eq:error-sum}
\end{align}
In particular, every $\delta_t<1$.  Iterating \eqref{eq:step-inclusion} gives
\[
  \left ( \prod_{t = 0}^{T-1}(1-\delta_t) \right )\cdot h_{C^{(t)}}(u)
  \leq h_{C}(u)
  \leq \left ( \prod_{t = 0}^{T-1}(1+\delta_t) \right )\cdot h_{C^{(t)}}(u).
\]
Since $\prod_t(1-\delta_t)\ge 1-\sum_t\delta_t$,
\[
  \prod_{t=0}^{T-1}(1-\delta_t)
  \ge 1-\eps.
\]
Likewise, by \eqref{eq:error-sum},
\[
  \prod_{t=0}^{T-1}(1+\delta_t)
  \le \exp\left(\sum_{t=0}^{T-1}\delta_t\right)
  \le e^{\eps/2}
  \le 1+\eps,
\]
where the last inequality holds for $0\le\eps\le 1$.

All together, this then implies that for every $u$, 
\[
  (1-\eps) \cdot h_C(u)
  \subseteq h_{C^{(T)}}(u)
  \subseteq (1+\eps) h_C(u).
\]
Finally,
\[
  C^{(T)}=\sum_{i=1}^m w_i^{(T)}C_i
\]
is supported on $m_T\le M=O(d/\eps^2)$ indices.  Taking $\lambda_i=w_i^{(T)}$ completes the proof.
\end{proof}

\section{Deriving Seminorm Sparsifiers}\label{sec:sparsifiersSemiNorm}

In this section, we now adapt the results from the previous section to building \emph{seminorm} sparsifiers. We first recall the definition of a seminorm:

\begin{definition}\label{def:seminorms}
    A function $f: \R^n \rightarrow \R_{\geq 0}$ is said to be a seminorm if it satisfies the following properties:
    \begin{enumerate}[(1)]
        \item \parhead{Positive Homogeneity} For all $x \in \R^n$ and $\lambda \in \R$, $f(\lambda x) = |\lambda| \cdot f(x)$.
        \item \parhead{Subadditivity} For all $x, y \in \R^n$, $f(x + y) \leq f(x) + f(y)$.
\end{enumerate}
\end{definition}

Importantly, we have the following claim which relates seminorms to support functions of centrally symmetric convex bodies:

\begin{fact}[Theorem 1.7.1 of \cite{Schneider14}]\label{fact:seminormImpliesConvex}
    Let $f: \R^n \rightarrow \R_{\geq 0}$ be a seminorm. Then, there exists a compact, convex, centrally symmetric convex set $C \subseteq \R^n$ such that
    \[
    f(x) = h_{C}(x).
    \]
\end{fact}

\begin{proof}[Remark.]
Note that \cite{Schneider14} Theorem 1.7.1 uses the notion of ``sublinear functions'' which refer to functions that satisfy both positive homogeneity and subaddativity, thereby including seminorms (see \cref{def:seminorms}). Theorem 1.7.1 does not explicitly say that $C \subseteq \R^n$ is centrally symmetric in the case of seminorms, but this follows from the proof \cite[p.46]{Schneider14} which shows that
\[
C := \{x \in \R^n \mid \langle x, u\rangle \le f(u)\text{ for every }u \in \R^n\}
\]
which is centrally symmetric.
\end{proof}

With this, we can now conclude with the statement of sparsifying sums of seminorms:

\begin{theorem}\label{thm:seminorm-sparsifiers}
    Let $f_1, \dots , f_m: \R^n \rightarrow \R_{\geq 0}$ be a collection of seminorms. Then, for any $\eps > 0$, there exists a set of weights $\lambda_1, \dots , \lambda_m \in \R_{\geq 0}$ with only $O(n / \eps^2)$ of the weights non-zero such that for every $x \in \R^n$:
    \[
    (1 - \eps) \cdot \sum_{i = 1}^m  f_i(x) \leq \sum_{i = 1}^m \lambda_i \cdot f_i(x) \leq  (1 + \eps) \cdot \sum_{i = 1}^m  f_i(x).
    \]
\end{theorem}

\begin{proof}
    Indeed, for each seminorm $f_i$, let $C_i$ denote the corresponding centrally symmetric, convex set as guaranteed by \cref{fact:seminormImpliesConvex}. It is then the case that for every $x \in \R^n$ and every $i \in [m]$ that 
    \[
    h_{C_i}(x) = f_i(x).
    \]

    We can then invoke \cref{thm:body-main} on the collection $C_1, \dots , C_m$. Note that the dimension of the span of these bodies is bounded by $n$, as they are contained in $\R^n$. This guarantees $\lambda_1, \dots , \lambda_m \in \R_{\geq 0}$ such that at most $O(n / \eps^2)$ of the weights are non-zero, and yet still, for every $u \in \R^n$
    \[
    (1-\eps) \cdot \sum_{i = 1}^m h_{C_i}(u)
  \leq\sum_{i = 1}^m \lambda_i \cdot h_{C_i}(u)
  \leq (1+\eps) \cdot \sum_{i = 1}^m h_{C_i}(u).
    \]
    This then immediately implies that the same weights $\lambda_1, \dots , \lambda_m$ also satisfy 
    \[
    (1 - \eps) \cdot \sum_{i = 1}^m  \cdot f_i(x) \leq \sum_{i = 1}^m \lambda_i \cdot f_i(x) \leq  (1 + \eps) \cdot \sum_{i = 1}^m  \cdot f_i(x)
    \]
    for every $x \in \R^n$. This concludes the proof.
\end{proof}

\section{Applications}\label{sec:applications}

With \cref{thm:seminorm-sparsifiers} in hand, we are now able to derive the sparsifiability of many concrete objects. In this section, we discuss two such applications, both improving over the state of the art. 

\subsection{Hypergraph Cut Sparsifiers}

To start, we consider the setting of \emph{hypergraph cut sparsifiers}. In this setting, one is given a hypergraph $H$ and a hypergraph cut-sparsifier is a re-weighted subset of hyperedges which preserves all cut-sizes:

\begin{definition}\label{def:hypergraphCutSparsifiers}
    Let $H = (V, E, w)$ be a hypergraph, where $w: E \rightarrow \R_{\geq 0}$ is a set of non-negative edge-weights. For a set $S \subseteq V$, the cut-size of $S$ is defined as
    \[
    \delta_H(S) := \sum_{e \in E} w(e) \cdot \mathbf{1}[S \cap e \neq \emptyset \text{ and } (V-S) \cap e \neq \emptyset].
    \]
    For a parameter $\eps > 0$, hypergraph $H'$ is said to be a $(1 \pm \eps)$ cut-sparsifier of $H$ if for every $S \subseteq V$ it is the case that 
    \[
    (1 - \eps) \delta_H(S) \leq \delta_{H'}(S) \leq (1 + \eps) \delta_H(S).
    \]
\end{definition}

We now define an auxiliary relaxation of the hypergraph cut function:

\begin{definition}\label{def:relaxedHyperedge}
    Let $n$ be an integer, and let $e \subseteq [n]$. For $x \in \R^n$, define $f_e(x) = \max_{u, v \in e} |x_u - x_v|$.
\end{definition}

Importantly, we have the following immediate consequences:

\begin{claim}\label{clm:hypergraphCutSeminorm}
     Let $n$ be an integer, and let $e \subseteq [n]$. Then, for $f_e$ as defined in \cref{def:relaxedHyperedge}, we have that:
     \begin{enumerate}[(1)]
         \item $f_e$ is a seminorm.
         \item For any $S \subseteq [n]$, $f_e(\mathbf{1}_S) = \mathbf{1}[S \cap e \neq \emptyset \text{ and } ([n]-S) \cap e \neq \emptyset]$.
     \end{enumerate}
\end{claim}

\begin{proof}
    Indeed, we can see that $f_e(x) \geq 0$. Likewise, for any $\lambda \in \R$, $f_e(\lambda \cdot x) = \max_{u, v \in e} |\lambda \cdot x_u - \lambda \cdot x_v| = |\lambda| \cdot \max_{u, v \in e} |x_u - x_v| = |\lambda| \cdot f_e(x)$. For subaddativity, we can observe that 
    \[
    f_e(x + y) = \max_{u, v \in e} |(x+y)_u - (x+y)_v| \leq \max_{u, v \in e} \left ( |x_u - x_v| + |y_u - y_v| \right ) \leq f_e(x) + f_e(y).
    \]

    Now, for the second condition, we can observe that
    \[
    f_e(\mathbf{1}_S) = \max_{u, v \in e} |(\mathbf{1}_S)_u - (\mathbf{1}_S)_v|.
    \]
    This expression evaluates to $1$ if and only if there is some choice of $u, v$ such that $u \in e \cap S$, $v \in e \cap ([n] - S)$, and otherwise the expression equals $0$. This exactly matches $\mathbf{1}[S \cap e \neq \emptyset \text{ and } ([n]-S) \cap e \neq \emptyset]$.
\end{proof}

With this established, we now have the following theorem:

\begin{theorem}
    Let $H = (V, E, w)$ be any hypergraph (with non-negative weights). Then, for any $\eps > 0$, there exists a hypergraph $H' = (V, E, w')$ (again, with non-negative weights) such that:
    \begin{enumerate}[(1)]
        \item $H'$ is a $(1 \pm \eps)$ hypergraph cut-sparsifier of $H$. 
        \item $H'$ has at most $O(|V| / \eps^2)$ hyperedges with non-zero weight: i.e., $ \left | \{ e \in E: w'(e) > 0\}\right | \leq O(|V| / \eps^2)$.
    \end{enumerate}
\end{theorem}

\begin{proof}
    We consider the collection of seminorms $w(e) \cdot f_e: \R^V \rightarrow \R_{\geq 0}$ for $f_e$ as defined in \cref{def:relaxedHyperedge}. Letting $m$ denote the number of hyperedges, by invoking \cref{thm:seminorm-sparsifiers}, we immediately obtain that there exists a set of weights $\lambda_1, \dots , \lambda_m \in \R_{\geq 0}$ with only $O(|V| / \eps^2)$ of the weights non-zero such that for every $x \in \R^V$:
    \[
    (1 - \eps) \cdot \sum_{e \in E}  w(e) f_e(x) \leq \sum_{e \in E} \lambda_e \cdot w(e) \cdot f_e(x) \leq  (1 + \eps) \cdot \sum_{e \in E}  w(e) f_e(x).
    \]

    Importantly, this implies that for every $S \subseteq V$ that 
    \[
    (1 - \eps) \cdot \sum_{e \in E}  w(e) f_e(\mathbf{1}_S) \leq \sum_{e \in E} \lambda_e \cdot w(e) \cdot f_e(\mathbf{1}_S) \leq  (1 + \eps) \cdot \sum_{e \in E}  w(e) f_e(\mathbf{1}_S).
    \]

    By the equivalence of \cref{clm:hypergraphCutSeminorm}, we then see that the preceding equation is equivalent to guaranteeing that for every $S \subseteq V$,
    \[
    (1 - \eps) \delta_H(S) \leq \delta_{H'}(S) \leq (1 + \eps) \delta_H(S),
    \]
    where $H'$ is the hypergraph with edge set $E$ and weights $w'(e) = \lambda(e) \cdot w(e)$. This then concludes the claim, as at most $O(|V| / \eps^2)$ edges have non-zero weight under $w'$.
\end{proof}

\subsection{Symmetric Submodular Sparsification}

Following the work of \cite{jambulapati2023sparsifying}, in this section, we show how \cref{thm:seminorm-sparsifiers} leads to better sparsifiers for sums of symmetric submodular functions. 

Recall that for a set $[n]$, the function $f: 2^{[n]} \rightarrow \R_{\geq 0}$ is said to be submodular if for all $S, T \subseteq [n]$, 
\[
f(S) + f(T) \geq f(S \cap T) + f(S \cup T).
\]
Such a submodular function is additionally said to be \emph{symmetric} if $f(S) = f([n] - S)$ for every $S \subseteq V$.

We now define sparsification for sums of symmetric submodular functions:
\begin{definition}
    Let $n$ be an integer, and let $f_1, \dots , f_m: 2^{[n]} \rightarrow \R_{\geq 0}$ be symmetric submodular functions. For a parameter $\eps > 0$, we say that a collection of weights $\lambda_1, \dots , \lambda_m \in \R_{\geq 0}$ is a $(1 \pm \eps)$ sparsifier of $f_1, \dots , f_m$ if for every $S \subseteq [n]$ it is the case that 
    \[
    (1 -\eps) \cdot \sum_{i = 1}^m f_i(S) \leq \sum_{i =1}^m \lambda_i \cdot f_i(S) \leq (1 + \eps) \cdot \sum_{i = 1}^m f_i(S).
    \]
\end{definition}

As a direct consequence of \cref{thm:seminorm-sparsifiers}, we obtain the following:

\begin{theorem}\label{thm:symmetricSparsification}
    Let $n$ be an integer, and let $f_1, \dots , f_m: 2^{[n]} \rightarrow \R_{\geq 0}$ be symmetric submodular functions. For any parameter $\eps > 0$, there exists a collection of weights $\lambda_1, \dots , \lambda_m \in \R_{\geq 0}$ such that:
    \begin{enumerate}[(1)]
        \item $\lambda_1, \dots , \lambda_m$ constitutes a $(1 \pm \eps)$ sparsifier of $f_1, \dots , f_m$.
        \item At most $O(n / \eps^2)$ of the weights $\lambda_1, \dots , \lambda_m$ are non-zero. 
    \end{enumerate}
\end{theorem}

The key ingredient in proving \cref{thm:symmetricSparsification} (paralleling the key ingredient of the analogous statement in \cite{jambulapati2023sparsifying}) is the following:

\begin{fact}[See, for instance, \cite{lovasz1983submodular, jambulapati2023sparsifying}]\label{fact:lovaszExtension}
    Let $n$ be an integer and let $f: 2^{[n]} \rightarrow \R_{\geq 0}$ be a symmetric submodular function. Then, the Lov\'asz extension of $f$, denoted $\hat{f}: \R^n \rightarrow \R_{\geq 0}$ is a continuous seminorm such that for any $S \subseteq [n]$, $f(S) = \hat{f}(\mathbf{1}_S)$.
\end{fact}

With this fact, we can now prove the stated theorem.

\begin{proof}[Proof of \cref{thm:symmetricSparsification}.]
    For each symmetric submodular function $f_i$, we replace it with its corresponding Lov\'asz extension $\hat{f_i}$. We then invoke \cref{thm:seminorm-sparsifiers} on the collection $\hat{f_i}: i \in [m]$. This returns $\lambda_1, \dots , \lambda_m \in \R_{\geq 0}$ with only $O(n / \eps^2)$ of the weights non-zero such that for every $x \in \R^n$:
    \[
    (1 - \eps) \cdot \sum_{i = 1}^m \hat{f}_i(x) \leq \sum_{i = 1}^m \lambda_i \cdot \hat{f}_i(x) \leq  (1 + \eps) \cdot \sum_{i = 1}^m \hat{f}_i(x).
    \]
    In particular, by taking $x = \mathbf{1}_S$ for any $S \subseteq [n]$, this implies that the same set of weights $\lambda_1, \dots , \lambda_m \in \R_{\geq 0}$ is also a $(1 \pm \eps)$ sparsifier of the collection $f_1, \dots , f_m$. This concludes the theorem. 
\end{proof}

As an immediate corollary, this also implies $O(n / \eps^2)$-size sparsifiers for sums of \emph{monotone} submodular functions:

\begin{definition}
    A submodular function $f: 2^{[n]} \rightarrow \R_{\geq 0}$ is said to be monotone if for any $S \subseteq [n]$ and $v \in [n] - S$, 
    \[
    f(S \cup \{v \}) \geq f(S).
    \]
\end{definition}

We then have the following corollary:

\begin{corollary}
Let $n$ be an integer, and let $f_1, \dots , f_m: 2^{[n]} \rightarrow \R_{\geq 0}$ be monotone submodular functions. For any parameter $\eps > 0$, there exists a collection of weights $\lambda_1, \dots , \lambda_m \in \R_{\geq 0}$ such that:
    \begin{enumerate}[(1)]
        \item $\lambda_1, \dots , \lambda_m$ constitutes a $(1 \pm \eps)$ sparsifier of $f_1, \dots , f_m$.
        \item At most $O(n / \eps^2)$ of the weights $\lambda_1, \dots , \lambda_m$ are non-zero. 
    \end{enumerate}
\end{corollary}

\begin{proof}
    This follows from Theorem 1.5 of \cite{khanna2024almost}, which states that for any collection of monotone submodular functions $f_1, \dots , f_m: 2^{[n]} \rightarrow \R_{\geq 0}$, one can define a collection of symmetric submodular functions $\tilde{f_1}, \dots , \tilde{f_m}: 2^{[n] \cup \{ \star\}}$ where for every $i \in [m]$ and $S \subseteq [n]$:
    \[
    f_i(S)= \tilde{f_i}(S) = \tilde{f_i}([n] \cup \{ \star\} - S).
    \]
    Sparsifying this collection of $\tilde{f_1}, \dots , \tilde{f_m}$ then immediately implies sparsification of $f_1, \dots , f_m$.
\end{proof}

\bibliographystyle{alphaurl}
\bibliography{convex}

\appendix

\end{document}